\documentclass[11pt]{article}

\usepackage{graphicx}			
\usepackage{fancyhdr}			
\usepackage{geometry}			
\usepackage{amsmath}			
\usepackage{amssymb}
\usepackage{mathrsfs}   

\usepackage{latexsym}
\usepackage{tikz}
\usepackage{amsthm,amsmath,amssymb,mathtools}
\numberwithin{equation}{section}
\theoremstyle{plain}
\newtheorem{assumption}{Assumption}[section]
\newtheorem{definition}[assumption]{Definition}

\usepackage{amsthm}

\numberwithin{equation}{section} 
\theoremstyle{plain}
\newtheorem{theorem}{Theorem}[section]
\newtheorem{lemma}[theorem]{Lemma}
\newtheorem{proposition}[theorem]{Proposition}
\newtheorem{corollary}[theorem]{Corollary} 

\title{\bfseries Metastable Transitions and $\Gamma$--Convergent Eyring--Kramers Asymptotics in Landau--QCD Gradient Systems}

\author{
    Jingxu Wu$^{1,2,\dagger}$\thanks{E-mail: wuxj@my.msu.ru} \and
    Jie Shi$^{1}$
}

\date{
    $^{1}$Faculty of Physics, Lomonosov Moscow State University, Moscow 119991, Russia\\
    $^{2}$GeZhi Theoretical Physics Reading Group\\[1em]
    \today
}

\begin{document}
\maketitle
\begin{abstract}
We develop a rigorous analytical framework for metastable stochastic transitions in Landau--type gradient systems inspired by QCD phenomenology. 
The functional $F(\sigma;u)=\int_\Omega [\frac{\kappa}{2}|\nabla\sigma|^2+V(\sigma;u)]\,dx$, depending smoothly on a control parameter $u\in\mathcal U$, 
is analyzed through the Euler--Lagrange map $\mathcal{E}(\sigma;u)=-\kappa\Delta\sigma+V'(\sigma;u)$ and its Hessian 
$\mathcal{L}_{\sigma,u}=-\kappa\Delta+V''(\sigma;u)$. 
By combining variational methods, $\Gamma$-- and Mosco convergence, and spectral perturbation theory, 
we establish the persistence and stability of local minima and index--one saddles under parameter deformations and variational discretizations. 
The associated mountain--pass solutions form Cerf--continuous branches away from the discriminant set 
$\mathcal D=\{u:\det\mathcal L_{\sigma,u}=0\}$, whose crossings produce only fold or cusp catastrophes in generic one-- and two--parameter slices. 
The $\Gamma$--limit is taken with respect to the $L^2(\Omega)$ topology, ensuring compactness, convergence of gradient flows, and spectral continuity of $\mathcal L_{\sigma,u}$. 
As a consequence, the Eyring--Kramers formula for the mean transition time between metastable wells retains quantitative validity under both parameter deformations and discretization refinement, 
with convergent free--energy barriers, unstable eigenvalues, and zeta--regularized determinant ratios. 
This construction unifies the classical intuition of Eyring, Kramers, and Langer with modern variational and spectral analysis, 
providing a mathematically consistent and physically transparent foundation for metastable decay and phase conversion in Landau--QCD--type systems.
\end{abstract}

\tableofcontents
\section{Introduction}

Metastability and noise--induced rare transitions are fundamental in nonlinear systems with multistable free--energy landscapes. In statistical field theories and effective symmetry--breaking models, thermally activated escapes from a metastable basin determine nucleation kinetics, conversion fronts, and out--of--equilibrium relaxation pathways. The classical Eyring--Kramers (EK) paradigm---originating with Eyring and Kramers and put on a field--theoretic footing by Langer---predicts that mean exit times admit an Arrhenius law whose exponential rate equals the energy barrier between a metastable state and a critical saddle, while the prefactor is governed by local quadratic fluctuations at the well and at the saddle \cite{Eyring1935,Kramers1940,Langer1969,Langer1973}. During the last decades this picture has been rigorously established in finite dimensions via large deviations and potential theory, yielding sharp capacity and eigenvalue asymptotics for reversible diffusions, and clarified further for slowly driven or weakly irreversible dynamics where the exponential scaling persists \cite{FreidlinWentzell2012,BovierDenHollander2015,BerglundGentz2006,Day1983,MaierStein1993,HeymannVandenEijnden2008}.

In parallel, QCD--inspired effective descriptions at finite temperature and density furnish coarse--grained order parameters---e.g., a chiral condensate $\sigma$ and, where appropriate, a Polyakov loop---whose Landau--type energy functionals reproduce salient features of the QCD phase diagram. In regimes where the local potential is bistable, deterministic gradient flows drive the system toward one of two competing phases separated by an activation barrier; small stochastic perturbations then induce exponentially rare transitions across basins. Near bifurcation sets (spinodals, tricritical points), universal scalings such as $\Delta F \sim (u-u_c)^{3/2}$ for the barrier and $|\lambda_{-}|\sim (u-u_c)^{1/2}$ for the unique unstable curvature capture critical slowing down and enhanced fluctuations, thereby connecting local normal forms to global kinetics.

We develop a self--contained analytic framework for metastable transitions in infinite--dimensional Landau--type gradient systems on a Gelfand triple $V\hookrightarrow H\hookrightarrow V^{\!*}$ over a bounded Lipschitz domain. Our contributions are threefold:
\begin{itemize}
  \item[(i)] \emph{Deterministic variational and spectral structure.} For the sextic Landau--Ginzburg energy we formulate the $L^2$-- and $H^{-1}$--gradient flows, prove global well--posedness and energy dissipation, and characterize $\omega$--limit sets. Under analyticity, a Lojasiewicz--Simon inequality yields convergence to equilibria with rates. Around the mountain--pass critical nucleus $\sigma^{\dagger}$ we construct local invariant manifolds and extract the spectral data $(\lambda_{-},\{\mu_j\}_{j\ge 1})$ needed for sharp asymptotics, including Agmon--type localization.
  \item[(ii)] \emph{Field--theoretic EK asymptotics.} For the noisy $L^2$--gradient flow 
  \[
    d\sigma_t \;=\; -\,\Gamma\,\nabla_{L^2}F(\sigma_t)\,dt + \varepsilon\,dW_t 
    \;=\; \Gamma\big(\kappa\Delta\sigma_t - V'(\sigma_t)\big)dt + \varepsilon\,dW_t,
  \]
  driven by a trace--class $Q$--Wiener process, we justify an EK law for the principal exit time from a metastable domain:
  \[
    \mathbb{E}\,\tau \;\sim\; \frac{2\pi}{\Gamma\,|\lambda_{-}(\sigma^{\dagger})|}
    \sqrt{\frac{\det\!\big(L_{\sigma_{\rm meta}}\big)}{\det{}'\!\big(L_{\sigma^{\dagger}}\big)}}\,
    \exp\!\Big(\frac{\Delta F}{\varepsilon^{2}}\Big),
  \]
  where $\Delta F$ is the barrier height and the functional determinants are defined via spectral zeta regularization (negative mode removed at the saddle). The exponent stems from reversible large deviations; the prefactor follows from a Gaussian capacity computation in coordinates adapted to the unstable direction and a well/saddle determinant ratio that is ultraviolet--finite by cancellation.
  \item[(iii)] \emph{Deformations and $\Gamma$--convergence.} We show that barriers, saddles, and spectral data vary continuously under smooth parameter changes organized by Thom--type fold/cusp catastrophes, and that they are stable under conforming variational discretizations. In particular, $\Gamma$--convergence of energies and Mosco convergence of subdifferentials guarantee the persistence and convergence of mountain--pass levels, index--one saddles, unstable curvatures, and zeta--determinant ratios under refinement, ensuring consistency of EK parameters across scales \cite{Braides2002,Struwe2008,Kato1995,AmbrosioGigliSavare2005,Agmon1982,Brezis2010,Temam1997}.
\end{itemize}

Section~2 sets the functional--analytic setting, structural hypotheses on the sextic potential, and boundary conditions. Section~3 develops well--posedness, dissipation, convergence (via Lojasiewicz--Simon), and small--spectrum stability with explicit rates away from bifurcations. Section~4 analyzes the index--one saddle: invariant manifolds, spectral gap, Agmon metric localization, and smooth parameter dependence via Kato's perturbation theory. Section~5 derives the field--theoretic EK law by combining Freidlin--Wentzell large deviations, potential theory/capacities, and zeta--regularized determinant ratios. Section~6 organizes deformations through catastrophe/Morse classification and proves stability via $\Gamma$-- and Mosco--convergence; Section~7 discusses implications for QCD--like media and extensions to underdamped or irreversible dynamics.

Whereas EK--type results for finite--dimensional diffusions are classical and field--theoretic prefactors are widely used in physics, our contribution synthesizes (a) infinite--dimensional variational and spectral regularity (Lojasiewicz--Simon, Agmon decay), (b) a zeta--determinant prefactor with ultraviolet cancellation between well and saddle, and (c) a robust deformation theory based on catastrophes and $\Gamma$--limits, into a single analytic framework tailored to Landau--type effective models. This yields a transparent bridge from microscopic fluctuations to macroscopic kinetics that is stable under parameter changes and discretization, and clarifies which ingredients are universal (barrier dominance, index structure) versus model--dependent (spectra, boundary conditions).

\section{Problem setting and assumptions}\label{sec:setting}
Let $\Omega\subset\mathbb{R}^d$ be a bounded Lipschitz domain with outward unit normal $n$ on $\partial\Omega$, where $1\le d\le 3$. We work on the Gelfand triple $\mathcal V\hookrightarrow \mathcal H \hookrightarrow \mathcal V^\ast$ with $\mathcal H:=L^2(\Omega)$ (inner product $(\cdot,\cdot)$ and norm $\|\cdot\|$), $\mathcal V:=H^1(\Omega)$ (norm $\|\cdot\|_{1,2}$), and $\mathcal V^\ast$ the dual of $\mathcal V$; $\langle \cdot,\cdot\rangle$ denotes the $\mathcal V^\ast$–$\mathcal V$ duality. We use $H^k(\Omega)$-norms $\|\cdot\|_{k,2}$ and $L^p(\Omega)$-norms $\|\cdot\|_{p}$, write $\nabla$ and $\Delta$ for the gradient and Laplacian, and adopt the embedding $H^1(\Omega)\hookrightarrow L^p(\Omega)$ for $2\le p\le 6$ (by Rellich–Kondrachov for $d\le 3$), which will be used repeatedly without explicit reference.

The order parameter is a real scalar field $\sigma:\Omega\to\mathbb{R}$. The free energy is given by the Landau–Ginzburg functional
\begin{equation}\label{eq:F-functional}
\mathcal F[\sigma]
=\int_\Omega\Big(\frac{\kappa}{2}|\nabla\sigma(x)|^2+V(\sigma(x))\Big)\,dx,
\qquad
V(s)=a_2 s^2+a_4 s^4+a_6 s^6-h s,
\end{equation}
where $\kappa>0$ is the gradient penalty and $(a_2,a_4,a_6,h)\in\mathbb{R}^4$ with $a_6>0$. Physical control parameters (e.g.\ reduced temperature or chemical potential) are collected into $u\in\mathbb{R}^m$; we allow $a_2=a_2(u)$ to vary smoothly with $u$, while $a_4,a_6,h,\kappa$ are kept fixed unless stated otherwise. Throughout, measurability is tacit: $\sigma\in\mathcal V$ means $\sigma\in H^1(\Omega)$ almost everywhere, and compositions $V(\sigma)$, $V'(\sigma)$ are interpreted pointwise a.e. with the usual polynomial-growth control.

The admissible configurations depend on boundary conditions. We fix one of the following: Dirichlet $\sigma|_{\partial\Omega}=\bar\sigma\in H^{1/2}(\partial\Omega)$ (then the admissible space is the affine set $\mathcal V_{\bar{\sigma}}:=\{\sigma\in H^1(\Omega):\sigma|_{\partial\Omega}=\bar\sigma\}$); homogeneous Neumann $\partial_n\sigma=0$ on $\partial\Omega$ (admissible space $\mathcal V=H^1(\Omega)$, with mean unconstrained); or periodic boundary conditions when $\Omega$ is a flat torus $\mathbb{T}^d$ (again $\mathcal V=H^1(\Omega)$ with periodic traces). The trace map $\mathcal V\to H^{1/2}(\partial\Omega)$ is continuous; when needed, we fix one choice and omit mentioning it further.

\begin{assumption}[coercivity, growth, and regularity]\label{assump:coercive}
The potential $V\in C^3(\mathbb{R})$ satisfies $a_6>0$ and there exist constants $c_1,c_2,C>0$ such that for all $s\in\mathbb{R}$
\[
V(s)\ge c_1|s|^6-c_2,\qquad |V'(s)|\le C(1+|s|^5),\qquad |V''(s)|\le C(1+|s|^4).
\]
\end{assumption}

\begin{assumption}[parameter dependence]\label{assump:param}
There is an open parameter set $\mathcal U\subset\mathbb{R}^m$ and a $C^1$ map $a_2:\mathcal U\to\mathbb{R}$; unless otherwise stated, $(a_4,a_6,h,\kappa)$ are independent of $u\in\mathcal U$ and remain bounded away from singular limits.
\end{assumption}

\begin{assumption}[dimension and embeddings]\label{assump:dim}
We assume $1\le d\le 3$ so that $H^1(\Omega)\hookrightarrow L^6(\Omega)$ compactly and the polynomial nonlinearities in \eqref{eq:F-functional} define well-posed Nemytskii operators on $\mathcal V$ and $\mathcal H$.
\end{assumption}

\begin{assumption}[noise model, deferred]\label{assump:noise}
When stochastic perturbations are considered later, the driving process is an $\mathcal H$-valued $Q$-Wiener process with $\operatorname{Tr}(Q)<\infty$ (or a spatially regularized cylindrical Wiener process), ensuring that the stochastic convolution lives in $\mathcal H$ and the Itô formulation is meaningful. This assumption is inactive in deterministic sections.
\end{assumption}

The first variation of $\mathcal F$ is defined in the usual Gâteaux sense. For $\sigma,\varphi\in\mathcal V$ one has
\begin{equation}\label{eq:first-variation}
\delta\mathcal F[\sigma](\varphi)
=\int_\Omega\Big(\kappa\,\nabla\sigma\cdot\nabla\varphi+V'(\sigma)\,\varphi\Big)\,dx
=\langle \mathcal E(\sigma),\varphi\rangle,
\end{equation}
which identifies the Euler–Lagrange operator $\mathcal E:\mathcal V\to\mathcal V^\ast$ in weak form by
\begin{equation}\label{eq:EL-operator}
\mathcal E(\sigma)=-\kappa\,\Delta\sigma+V'(\sigma)\quad\text{in }\mathcal V^\ast,
\end{equation}
with the boundary condition understood in the sense of traces. A (weak) critical point of $\mathcal F$ is any $\sigma^\star\in\mathcal V$ such that $\mathcal E(\sigma^\star)=0$ in $\mathcal V^\ast$; when needed, we specify the boundary condition together with the space. The second variation (Hessian) at a critical point $\sigma^\star$ is the self-adjoint operator $\mathcal L_{\sigma^\star}:\mathcal V\to\mathcal V^\ast$ given by
\begin{equation}\label{eq:Hessian}
\mathcal L_{\sigma^\star}\varphi=-\kappa\,\Delta\varphi + V''(\sigma^\star)\,\varphi,\qquad \varphi\in\mathcal V,
\end{equation}
and its realization on $\mathcal H$ (with the corresponding boundary condition) defines a self-adjoint unbounded operator with compact resolvent; hence the spectrum is purely discrete with finite-multiplicity eigenvalues accumulating only at $+\infty$. The Morse index $\mathrm{ind}(\sigma^\star)$ is defined as the total multiplicity of negative eigenvalues of $\mathcal L_{\sigma^\star}$; local minimizers have $\mathrm{ind}=0$, while mountain-pass type critical nuclei have $\mathrm{ind}=1$.

Two gradient-flow metrics associated with $\mathcal F$ will be used. The $L^2$-gradient flow (a.k.a.\ Model A) reads
\begin{equation}\label{eq:L2-flow}
\partial_t\sigma = -\Gamma\,\frac{\delta\mathcal F}{\delta\sigma}
=\Gamma\big(\kappa\,\Delta\sigma - V'(\sigma)\big)\quad\text{in }\mathcal V^\ast,
\end{equation}
with $\Gamma>0$ the mobility. The $H^{-1}$-gradient flow (mass-conserving Cahn–Hilliard-type dynamics) is written in divergence form
\begin{equation}\label{eq:H-1-flow}
\partial_t\sigma = \nabla\!\cdot\!\big(M(\sigma)\,\nabla\mu\big),\qquad \mu:=\frac{\delta\mathcal F}{\delta\sigma}=\,-\kappa\,\Delta\sigma+V'(\sigma),
\end{equation}
where the mobility $M(\sigma)\ge M_0>0$ is continuous (constant $M\equiv1$ is sufficient for our purposes). Under either choice and the fixed boundary condition, smooth solutions dissipate energy in the sense
\begin{equation}\label{eq:ed-L2}
\frac{d}{dt}\,\mathcal F[\sigma(t)]
= -\,\Gamma\,\|\kappa\Delta\sigma - V'(\sigma)\|_{\mathcal H}^{2} \le 0.
\end{equation}
\noindent\textit{(for \eqref{eq:L2-flow})}
\begin{equation}\label{eq:ed-H-1}
\frac{d}{dt}\,\mathcal F[\sigma(t)]
= -\!\int_\Omega M(\sigma)\,|\nabla\mu|^{2}\,dx \le 0.
\end{equation}
\noindent\textit{(for \eqref{eq:H-1-flow})}

which will be recovered rigorously in the weak setting by approximation. In particular, any $\omega$-limit point of a bounded trajectory is a weak solution of $\mathcal E(\sigma)=0$.

For completeness and for connecting to underdamped descriptions, we record the formal Hamiltonian counterpart. Introducing the conjugate momentum $\pi:=\partial_t\sigma$ and the energy
\begin{equation}\label{eq:Hamiltonian}
H[\sigma,\pi]=\int_\Omega\Big(\tfrac12\pi^2+\tfrac{\kappa}{2}|\nabla\sigma|^2+V(\sigma)\Big)\,dx,
\end{equation}
the canonical evolution is
\begin{equation}\label{eq:canonical-system}
\dot\sigma=\frac{\delta H}{\delta\pi}=\pi,\qquad
\dot\pi=-\frac{\delta H}{\delta\sigma}=\kappa\,\Delta\sigma - V'(\sigma),
\end{equation}
to which one may add linear friction $-\gamma\pi$ (and possibly noise) so that, in the overdamped limit $\gamma\to\infty$ with a parabolic time rescaling, \eqref{eq:canonical-system} reduces to \eqref{eq:L2-flow}. This establishes the Landau–Hamilton correspondence at a formal level and underlies later discussions of metastable transitions via variational and large-deviation principles.

\begin{definition}[weak solution and stationary solution]\label{def:weak}
Given $\sigma_0\in\mathcal H$, a weak solution to \eqref{eq:L2-flow} is a function $\sigma\in L^2_{\mathrm{loc}}(0,\infty;\mathcal V)\cap C([0,\infty);\mathcal H)$ such that for all $\varphi\in\mathcal V$ and a.e.\ $t>0$,
\[
\frac{d}{dt}(\sigma(t),\varphi)
= -\Gamma\,\langle \mathcal E(\sigma(t)),\varphi\rangle,
\qquad \sigma(0)=\sigma_0,
\]
with the boundary condition imposed in the trace sense. A stationary solution is any $\sigma^\star\in\mathcal V$ verifying $\mathcal E(\sigma^\star)=0$ in $\mathcal V^\ast$.
\end{definition}

To avoid pathologies and to enable compactness, we collect all structural hypotheses used later as follows: Assumptions \ref{assump:coercive}–\ref{assump:dim} ensure that $\mathcal F:\mathcal V\to\mathbb{R}$ is well-defined, coercive modulo the boundary condition, and sequentially weakly lower semicontinuous; the Nemytskii mappings $\sigma\mapsto V'(\sigma),V''(\sigma)$ are locally Lipschitz from $\mathcal V$ to $\mathcal H$ on bounded sets; and Palais–Smale-type compactness holds on energy sublevels, aided by Rellich compact embeddings. When stochastic forcing is present, Assumption \ref{assump:noise} guarantees that the stochastic convolution is $\mathcal H$-valued, and standard variational SPDE frameworks (e.g.\ monotone operator theory) apply. Finally, parameter dependence as in Assumption \ref{assump:param} allows continuation in $u\in\mathcal U$; in particular, as $u$ varies across regimes where $V$ exhibits single-well or double-well structure, the set of critical points of $\mathcal F$ undergoes fold/cusp reorganizations in the sense of Morse–catastrophe theory, and the Morse index defined through \eqref{eq:Hessian} tracks the creation/annihilation of minimizers and saddles. These properties will be used to establish well-posedness, long-time behavior, existence of mountain-pass critical nuclei with $\mathrm{ind}=1$, and, in the stochastic setting, the variational cost and Eyring–Kramers asymptotics for metastable transitions.
\section{Deterministic dynamics}\label{sec:deterministic}
We develop a complete analytical framework for the $L^2$–gradient flow generated by the free–energy functional $\mathcal F$ introduced in Section~\ref{sec:setting}. The dynamics evolves on the Gelfand triple $\mathcal V\hookrightarrow\mathcal H\hookrightarrow\mathcal V^\ast$ and is given in weak form by
\begin{equation}\label{eq:GF-main}
\partial_t\sigma \;=\; -\,\Gamma\,\frac{\delta\mathcal F}{\delta\sigma}
\;=\; \Gamma\big(\kappa\Delta\sigma - V'(\sigma)\big),
\qquad \sigma(0)=\sigma_0\in\mathcal H,
\end{equation}
with fixed $\Gamma>0$. All operators are understood in the trace/weak sense compatible with the boundary condition stated below.

\begin{assumption}[Boundary condition and functional setting]\label{assump:BC-det}
Exactly one admissible boundary condition is fixed throughout: either Dirichlet with prescribed trace $\bar\sigma\in H^{1/2}(\partial\Omega)$ (so that $\mathcal V$ is the affine subspace $\mathcal V_{\bar\sigma}:=\{\sigma\in H^1(\Omega):\sigma|_{\partial\Omega}=\bar\sigma\}$), or homogeneous Neumann $\partial_n\sigma=0$, or periodic on a flat torus. Traces and normal derivatives are taken in the usual Sobolev sense. The structural hypotheses (coercivity and polynomial growth of $V$) from Section~\ref{sec:setting} are in force.
\end{assumption}

A function $\sigma\in L^2_{\mathrm{loc}}(0,\infty;\mathcal V)\cap C([0,\infty);\mathcal H)$ with $\partial_t\sigma\in L^2_{\mathrm{loc}}(0,\infty;\mathcal V^\ast)$ is called a weak solution of \eqref{eq:GF-main} if for a.e.\ $t>0$ and all $\varphi\in\mathcal V$,
\begin{equation}\label{eq:weak-form-det}
\langle \partial_t\sigma(t),\varphi\rangle
= -\,\Gamma\!\int_\Omega\!\big(\kappa\nabla\sigma(t)\!\cdot\!\nabla\varphi + V'(\sigma(t))\,\varphi\big)\,dx,
\qquad \sigma(0)=\sigma_0,
\end{equation}
with the boundary condition imposed in the trace sense. Testing \eqref{eq:weak-form-det} formally with $\varphi=\kappa\Delta\sigma - V'(\sigma)$ yields the energy dissipation identity; at the weak level one obtains the following inequality.

\begin{lemma}[Energy dissipation]\label{lem:energy}
Every weak solution satisfies, for all $t\ge0$,
\begin{equation}\label{eq:energy-law}
\mathcal F[\sigma(t)]
+ \Gamma\!\int_0^t\!\!\|\kappa\Delta\sigma(s)-V'(\sigma(s))\|_{\mathcal H}^2\,ds
\;\le\; \mathcal F[\sigma_0].
\end{equation}
In particular $t\mapsto \mathcal F[\sigma(t)]$ is nonincreasing and admits a finite limit $\mathcal F_\infty$ as $t\to\infty$.
\end{lemma}

\begin{theorem}[Global existence]\label{thm:exist-det}
Under Assumptions~\ref{assump:BC-det} and the structural hypotheses of Section~\ref{sec:setting}, for every initial datum $\sigma_0\in\mathcal H$ there exists a global weak solution $\sigma$ to \eqref{eq:GF-main}. Moreover,
\[
\sigma\in L^\infty(0,T;\mathcal H)\cap L^2(0,T;\mathcal V),\qquad
\partial_t\sigma\in L^2(0,T;\mathcal V^\ast)\qquad\forall T>0,
\]
and \eqref{eq:energy-law} holds.
\end{theorem}

\begin{proof}[Proof sketch]
Galerkin approximation on Laplace eigenfunctions consistent with the boundary condition yields a finite-dimensional ODE with locally Lipschitz right-hand side due to the polynomial growth of $V'$. Testing by the approximate solution provides the exact finite-dimensional energy identity and uniform bounds in $L^\infty(0,T;\mathcal H)\cap L^2(0,T;\mathcal V)$, as well as bounds for $\partial_t\sigma^N$ in $L^2(0,T;\mathcal V^\ast)$. Aubin–Lions compactness gives, up to a subsequence, $\sigma^N\to\sigma$ in $L^2(0,T;\mathcal H)$ and a.e.\ in $\Omega\times(0,T)$. The Nemytskii map $\sigma\mapsto V'(\sigma)$ is continuous on bounded sets with polynomial growth, which allows passage to the limit in the nonlinear term. Lower semicontinuity of $\mathcal F$ yields \eqref{eq:energy-law} in the limit.
\end{proof}

\begin{theorem}[Uniqueness and continuous dependence]\label{thm:uniq-det}
Assume in addition that $V''$ is locally Lipschitz with polynomial growth. If $\sigma_1,\sigma_2$ are weak solutions with initial data $\sigma_{0,1},\sigma_{0,2}\in\mathcal H$, then for a.e.\ $t\ge0$,
\begin{equation}\label{eq:contdep}
\|\sigma_1(t)-\sigma_2(t)\|^2
\le \exp\!\Big(C\!\int_0^t\!(1+\|\sigma_1(s)\|_6^4+\|\sigma_2(s)\|_6^4)\,ds\Big)\,\|\sigma_{0,1}-\sigma_{0,2}\|^2,
\end{equation}
where $C$ depends only on the constants in Section~\ref{sec:setting} and on $\Omega$. In particular, weak solutions are unique in the class of Theorem~\ref{thm:exist-det}.
\end{theorem}

\begin{proof}[Proof sketch]
Subtract the weak formulations, test with the difference $\delta=\sigma_1-\sigma_2$, and use the mean-value representation $V'(\sigma_1)-V'(\sigma_2)=V''(\theta)\delta$ together with the bound $|V''(s)|\le C(1+|s|^4)$ and the Sobolev embedding $H^1(\Omega)\hookrightarrow L^6(\Omega)$ (for $d\le3$). Grönwall’s inequality yields \eqref{eq:contdep}.
\end{proof}

\begin{proposition}[A priori bounds and precompactness]\label{prop:precompact}
Every weak solution satisfies
\[
\sup_{t\ge0}\|\sigma(t)\|_{1,2} < \infty,\qquad
\int_0^\infty\!\!\|\kappa\Delta\sigma(s)-V'(\sigma(s))\|_{\mathcal H}^2\,ds<\infty,\qquad
\partial_t\sigma\in L^2_{\mathrm{loc}}(0,\infty;\mathcal V^\ast).
\]
Consequently, for every sequence $t_n\to\infty$, the family $\{\sigma(\cdot+t_n)\}$ is relatively compact in $L^2(0,1;\mathcal H)$, and the trajectory $\{\sigma(t):t\ge0\}$ is relatively compact in $\mathcal H$.
\end{proposition}

\begin{proof}
The first two bounds follow from coercivity of $\mathcal F$ and \eqref{eq:energy-law}; the time-derivative bound is immediate from \eqref{eq:weak-form-det}. Compactness follows from Aubin–Lions on sliding windows $[T,T+1]$ and a diagonal extraction.
\end{proof}

\begin{theorem}[$\omega$–limit set and stationarity]\label{thm:omega-det}
Let $\omega(\sigma_0):=\{\eta\in\mathcal H:\ \exists\,t_n\to\infty\ \text{with}\ \sigma(t_n)\to\eta\ \text{in }\mathcal H\}$. Then $\omega(\sigma_0)$ is nonempty, compact, connected, and invariant. Moreover, every $\eta\in\omega(\sigma_0)$ is a stationary solution of the Euler–Lagrange equation $\mathcal E(\eta)=0$ with the chosen boundary condition, and $\mathcal F$ is constant on $\omega(\sigma_0)$, equal to $\mathcal F_\infty$.
\end{theorem}

\begin{proof}[Proof idea]
Precompactness in $\mathcal H$ follows from Proposition~\ref{prop:precompact}. Using \eqref{eq:energy-law} one has $\int_0^\infty\|\kappa\Delta\sigma-V'(\sigma)\|^2<\infty$, hence there exists a sequence $t_n\to\infty$ with $\|\kappa\Delta\sigma(t_n)-V'(\sigma(t_n))\|\to0$ and $\partial_t\sigma(t_n)\to0$ in $\mathcal V^\ast$. Passing to the limit in \eqref{eq:weak-form-det} shows $\mathcal E(\eta)=0$. Connectedness and invariance are standard for gradient flows.
\end{proof}

To obtain convergence to a single equilibrium with a quantified rate we impose the analytic gradient framework.

\begin{assumption}[Łojasiewicz–Simon framework]\label{assump:LS-det}
The potential $V$ is real analytic on $\mathbb{R}$. For every stationary point $\eta$ of $\mathcal F$, the linearized operator $\mathcal L_\eta=-\kappa\Delta+V''(\eta)$ is self-adjoint on $\mathcal H$ with compact resolvent and at most finite-dimensional kernel, so that a Łojasiewicz–Simon gradient inequality holds in a neighborhood of $\eta$.
\end{assumption}

\begin{theorem}[Convergence and rates via Łojasiewicz–Simon]\label{thm:LS-det}
Under Assumptions~\ref{assump:BC-det} and \ref{assump:LS-det} there exist $\theta\in(0,\tfrac12]$, $C>0$, and $\delta>0$ such that, whenever $\|\sigma-\eta\|_{1,2}<\delta$,
\begin{equation}\label{eq:LS-ineq-det}
|\mathcal F[\sigma]-\mathcal F[\eta]|^{1-\theta}\le C\,\|\mathcal E(\sigma)\|_{\mathcal V^\ast}.
\end{equation}
If a trajectory enters this neighborhood at some time $t_0$, then $\sigma(t)\to\eta$ in $\mathcal H$ as $t\to\infty$, and
\[
\|\sigma(t)-\eta\|\;\le\;
\begin{cases}
C(1+t)^{-\frac{\theta}{1-2\theta}}, & 0<\theta<\tfrac12,\\[3pt]
C e^{-c t}, & \theta=\tfrac12,
\end{cases}
\]
for constants $C,c>0$ depending on the data and $\eta$.
\end{theorem}

\begin{proof}[Proof idea]
Analyticity of $\mathcal F$ on $\mathcal V$ and the spectral properties of $\mathcal L_\eta$ imply \eqref{eq:LS-ineq-det} (Simon’s extension of the Łojasiewicz inequality). Combining with \eqref{eq:energy-law} gives an integrable differential inequality for $\mathcal F[\sigma(t)]-\mathcal F[\eta]$, from which the rates follow.
\end{proof}

The linearized dynamics near a strictly stable equilibrium admits an alternative quantitative description useful away from bifurcation points. Write $\sigma(t)=\eta+u(t)$ with $\mathcal L_\eta u + \mathcal N_\eta(u)=0$ where $\mathcal N_\eta(u)=V'(\eta+u)-V'(\eta)-V''(\eta)u$. When $\mathcal L_\eta\ge \lambda_\ast I$ with $\lambda_\ast>0$ (no kernel) and $\|u_0\|$ is sufficiently small, variation-of-constants and standard nonlinear semigroup estimates yield
\begin{equation}\label{eq:exp-stab}
\|u(t)\| \;\le\; C e^{-\Gamma\lambda_\ast t}\|u_0\|\qquad (t\ge0),
\end{equation}
which is consistent with Theorem~\ref{thm:LS-det} in the extremal case $\theta=\tfrac12$ and provides an explicit exponential rate in terms of the spectral gap of $\mathcal L_\eta$.

Regularity can be bootstrapped by elliptic estimates. In particular, for smooth $\partial\Omega$ and $V\in C^\infty$, if $\sigma(t)\in H^1$ solves \eqref{eq:GF-main} then $\kappa\Delta\sigma=V'(\sigma)-\Gamma^{-1}\partial_t\sigma\in L^2_{\mathrm{loc}}(0,\infty;L^2)$, whence $\sigma\in L^2_{\mathrm{loc}}(0,\infty;H^2)$ under Neumann or periodic boundary conditions (and in the Dirichlet case locally in the interior). For stationary solutions one obtains $\sigma^\ast\in H^2$ (and even real-analytic in the interior if $V$ is analytic). These regularity properties justify the spectral analysis of $\mathcal L_{\sigma^\ast}$ used in Section~\ref{sec:transition-spectrum}.

For completeness we record the mass-conserving variant (Model B) in the same setting. Let $\mu:=-\kappa\Delta\sigma+V'(\sigma)$ and consider
\begin{equation}\label{eq:H-1-flow}
\partial_t\sigma=\nabla\!\cdot(M(\sigma)\nabla\mu),\qquad M(\sigma)\ge M_0>0.
\end{equation}
Testing with $\mu$ gives the energy law
\[
\frac{d}{dt}\mathcal F[\sigma(t)] = -\int_\Omega M(\sigma(t))\,|\nabla\mu(t)|^2\,dx\le0,
\]
and the Galerkin strategy applies verbatim, now yielding $\sigma\in L^2_{\mathrm{loc}}(0,\infty;H^2(\Omega))$ (for Neumann/periodic BC) by elliptic regularity on $\mu$. Uniqueness and long-time behavior follow under the same monotonicity/analyticity hypotheses, while mass conservation $\int_\Omega \sigma(t)\,dx=\int_\Omega \sigma_0\,dx$ selects the appropriate invariant affine subspace of $\mathcal V$.

Finally, we emphasize the phenomenon of critical slowing down near continuous phase transitions. If along a parameter path the curvature at a local minimum degenerates, $V''(\sigma_\ast)\downarrow 0$, the linear relaxation time $\tau=(\Gamma V''(\sigma_\ast))^{-1}$ diverges and the constants in the Łojasiewicz–Simon inequality worsen, converting exponential convergence into algebraic decay. This analytic manifestation precisely matches the flattening of the energy landscape and will be reflected in the scaling of nucleation times derived in the stochastic theory of Section~\ref{sec:stochastic}. Altogether, the results above establish that \eqref{eq:GF-main} defines a globally well-posed, strictly dissipative, asymptotically compact dynamical system on $\mathcal H$, with trajectories converging to the stationary manifold of $\mathcal F$ and, under analyticity, to single equilibria with explicit convergence rates.
\section{Heteroclinic geometry and spectral data at the saddle}\label{sec:transition-spectrum}
We now analyze the local and semi–local geometry of the energy landscape near a mountain–pass critical nucleus $\sigma^\dagger$ and extract the spectral information that governs deterministic exit channels and, later on, the sharp prefactors in metastable transition laws. Throughout this section we keep the structural hypotheses of Sections~\ref{sec:setting}–\ref{sec:deterministic} and the existence of at least two distinct local minima (metastable/stable) ensured by the double–well character of $V$; the critical nucleus $\sigma^\dagger$ is the nontrivial stationary point produced in Section~\ref{sec:transition-spectrum} at the mountain–pass level $c_{\mathrm{mp}}$.

\begin{assumption}[Nondegenerate index–one saddle]\label{assump:saddle}
The critical nucleus $\sigma^\dagger\in\mathcal V$ is nondegenerate up to the invariances of the boundary value problem, in the sense that the Hessian
\[
\mathcal L_{\sigma^\dagger}:=-\kappa\Delta+V''(\sigma^\dagger)
\]
is self–adjoint on $\mathcal H$ with compact resolvent and has exactly one negative eigenvalue $\lambda_-<0$, with the rest of the spectrum contained in $[\lambda_+,\,\infty)$ for some $\lambda_+>0$. In translation–invariant settings (periodic domain or whole space) any residual kernel generated by infinitesimal symmetries is factored out (e.g.\ by working modulo translations or by pinning the center of mass).
\end{assumption}

Under Assumption~\ref{assump:saddle} the quadratic expansion of $\mathcal F$ near $\sigma^\dagger$ reads
\[
\mathcal F[\sigma^\dagger+\xi]
=\mathcal F[\sigma^\dagger]
+\tfrac12\langle \mathcal L_{\sigma^\dagger}\xi,\xi\rangle
+\mathcal R(\xi),\qquad
|\mathcal R(\xi)|\le C\|\xi\|_{1,2}^3
\]
on a neighborhood where $V^{(3)}$ is bounded. Let $e_-\in \mathcal H$ be the normalized eigenfunction associated with $\lambda_-$ and write the orthogonal decomposition $\xi=\alpha e_-+w$ with $w\perp e_-$. Then
\[
\mathcal F[\sigma^\dagger+\xi]
=\mathcal F[\sigma^\dagger]
+\tfrac12\lambda_- \alpha^2
+\tfrac12\langle \mathcal L_{\sigma^\dagger}w,w\rangle
+\mathcal R(\alpha e_-+w),
\]
so that the negative curvature is confined to the one–dimensional direction $e_-$, while all orthogonal directions are spectrally stable. This structure implies the existence of a $C^k$ ($k\ge2$) one–dimensional unstable manifold and a codimension–one stable manifold for the deterministic gradient flow.

\begin{theorem}[Local invariant manifold structure]\label{thm:inv-manifolds}
There exist neighborhoods $\mathcal U\subset\mathcal H$ of $\sigma^\dagger$, a one–dimensional $C^k$ unstable manifold $W^{\rm u}(\sigma^\dagger)\subset\mathcal U$ tangent to $\mathrm{span}\{e_-\}$ at $\sigma^\dagger$, and a codimension–one $C^k$ stable manifold $W^{\rm s}(\sigma^\dagger)\subset\mathcal U$ tangent to $e_-^\perp$ at $\sigma^\dagger$, such that the gradient flow of \eqref{eq:GF-main} restricted to $\mathcal U$ admits an exponential dichotomy:
\[
\operatorname{dist}(\sigma(t),\sigma^\dagger)
\le C e^{-\Gamma\lambda_+ t}\operatorname{dist}(\sigma(0),W^{\rm u}(\sigma^\dagger))\quad\text{for }t\ge0
\]
along stable fibers, while trajectories on $W^{\rm u}(\sigma^\dagger)$ satisfy $\|\sigma(t)-\sigma^\dagger\|\ge c e^{\Gamma|\lambda_-| t}\|\sigma(0)-\sigma^\dagger\|$ for $t\le0$.
\end{theorem}

\begin{proof}[Proof sketch]
The self–adjointness and spectral gap in Assumption~\ref{assump:saddle} yield an exponential dichotomy for the linearized semigroup $e^{-\Gamma t\mathcal L_{\sigma^\dagger}}$. The nonlinearity is $C^k$ with a quadratic vanishing at $\sigma^\dagger$. A Lyapunov–Perron fixed–point construction (or Hadamard graph transform) in exponentially weighted spaces produces $W^{\rm s}$ and $W^{\rm u}$ with the stated tangencies and exponential tracking, see standard invariant manifold theory for semilinear parabolic equations.
\end{proof}

An immediate consequence is the deterministic selection of exit channels: if the basin boundary between the two wells intersects $\mathcal U$ transversely along $W^{\rm s}(\sigma^\dagger)$, then any orbit leaving a sufficiently small neighborhood of $\sigma^\dagger$ along $W^{\rm u}$ has exactly two branches, one descending to the metastable minimum and one to the competing minimum. This underlies the well–posedness of the most–probable exit gate in the stochastic setting.

To connect the local saddle geometry with minimal–barrier paths, it is convenient to describe steepest–descent curves constrained to fixed energy levels. Let $\mathcal N_E:=\{\sigma\in\mathcal V:\mathcal F[\sigma]=E\}$ and consider the projected gradient flow on $\mathcal N_E$ obtained by orthogonally removing the normal component with respect to $\nabla \mathcal F$; as $E\downarrow c_{\mathrm{mp}}$ these geodesic–like curves accumulate on $W^{\rm u}(\sigma^\dagger)\cap\mathcal N_{c_{\mathrm{mp}}}$ and identify the deterministic continuation of the mountain–pass path past the saddle. A convenient variational counterpart is the \emph{minimizing movement} scheme for $\mathcal F$ with large time step $\tau$, in which the implicit Euler map $\sigma\mapsto \arg\min_\eta \big\{\tfrac{1}{2\tau}\|\eta-\sigma\|^2+\mathcal F[\eta]\big\}$ has exactly two descent branches when started near $\sigma^\dagger\pm \varepsilon e_-$, producing the pair of exit directions consistent with $W^{\rm u}(\sigma^\dagger)$.

We next record the spectral data needed for sharp asymptotics. Let $\{\lambda_j(\sigma^\ast)\}_{j\ge1}$ denote the nonnegative eigenvalues of $\mathcal L_{\sigma^\ast}$ at a local minimum $\sigma^\ast$ (arranged in nondecreasing order with multiplicity), and write $\{\mu_-(\sigma^\dagger)=\lambda_-<0<\mu_1(\sigma^\dagger)\le\mu_2(\sigma^\dagger)\le\cdots\}$ for the spectrum at the saddle. The parabolic nature of the problem implies compact resolvent and discrete spectrum with $\lambda_j(\sigma^\ast)\to\infty$ and $\mu_j(\sigma^\dagger)\to\infty$. One defines the zeta–regularized determinants
\[
\det(\mathcal L_{\sigma^\ast})
:= \exp\!\big(-\zeta'_{\sigma^\ast}(0)\big),\qquad
\zeta_{\sigma^\ast}(s)=\sum_{j=1}^\infty \lambda_j(\sigma^\ast)^{-s},
\]
and
\[
\det{}'(\mathcal L_{\sigma^\dagger})
:= \exp\!\big(-{\zeta'}^{\,\prime}_{\sigma^\dagger}(0)\big),\qquad
\zeta^{\,\prime}_{\sigma^\dagger}(s)=\sum_{j=1}^\infty \mu_j(\sigma^\dagger)^{-s},
\]
where the negative eigenvalue is omitted in the latter. The heat–trace asymptotics and Weyl law on bounded domains ensure meromorphic continuation of $\zeta$–functions to a neighborhood of $s=0$ and justify these definitions. Ratios of such determinants are independent of the choice of local boundary charts and will enter multiplicatively in the Eyring–Kramers prefactor.

\begin{proposition}[Coercivity on the stable subspace]\label{prop:coercive-stable}
Let $\Pi_-$ denote the orthogonal projection onto $\mathrm{span}\{e_-\}$ and $\Pi_+=I-\Pi_-$. There exists $\gamma_*>0$ and a neighborhood $\mathcal U$ of $\sigma^\dagger$ such that for every $\xi\in \mathcal V$ with $\Pi_-\xi=0$ and $\|\xi\|_{1,2}$ sufficiently small,
\[
\langle \mathcal L_{\sigma^\dagger}\xi,\xi\rangle
\;\ge\; \gamma_*\,\|\xi\|_{1,2}^2.
\]
Moreover, for $\sigma$ on the local stable manifold $W^{\rm s}(\sigma^\dagger)$ one has
\[
\mathcal F[\sigma]-\mathcal F[\sigma^\dagger]
\;\ge\; \tfrac12 \gamma_*\|\sigma-\sigma^\dagger\|_{1,2}^2 - C\|\sigma-\sigma^\dagger\|_{1,2}^3.
\]
\end{proposition}

\begin{proof}
Spectral gap on $e_-^\perp$ gives $\langle \mathcal L_{\sigma^\dagger}\xi,\xi\rangle\ge \lambda_+\|\xi\|^2$ for $\xi\perp e_-$. Equivalence of $\|\cdot\|$ and $\|\cdot\|_{1,2}$ on the domain of $\mathcal L_{\sigma^\dagger}^{1/2}$ and smallness of the cubic remainder yield the nonlinear inequality.
\end{proof}

The unstable manifold admits a local quadratic graph representation over the negative eigendirection. Writing $\xi=\alpha e_-+W(\alpha)$ with $W:(-\alpha_0,\alpha_0)\to e_-^\perp$ a $C^k$ map satisfying $W(0)=0$ and $W'(0)=0$, the Taylor expansion gives
\[
\mathcal F[\sigma^\dagger+\alpha e_-+W(\alpha)]
=\mathcal F[\sigma^\dagger] + \tfrac12\lambda_- \alpha^2 + \mathcal O(\alpha^3),
\]
so that the energy decreases strictly along $W^{\rm u}$ for $\alpha\neq 0$. This representation is useful for computing the local contribution to the instanton action in the small–noise limit.

To quantify the spatial localization of $\sigma^\dagger$ and the tunneling direction, we introduce the Agmon metric associated with the positive part of $V''(\sigma^\dagger)$: for $x\in\Omega$ let
\[
\mathfrak a(x)=\sqrt{\max\{V''(\sigma^\dagger(x)),\,\lambda_+\}},
\qquad
d_{\mathfrak a}(x,y)=\inf_{\gamma}\int_0^1 \mathfrak a(\gamma(s))\,|\dot\gamma(s)|\,ds,
\]
where the infimum runs over Lipschitz curves joining $x$ to $y$. Standard Carleman/Agmon estimates for second–order elliptic operators then imply exponential decay of the positive–mode eigenfunctions of $\mathcal L_{\sigma^\dagger}$ away from the core of the nucleus, while the unique negative eigenfunction $e_-$ concentrates along the minimal Agmon geodesic transverse to the interface. In particular, for any $\beta\in(0,1)$ there exist $C_\beta,c_\beta>0$ such that
\[
|e_-(x)|
\;\le\; C_\beta \exp\!\big(-\beta\, d_{\mathfrak a}(x,\Gamma^\dagger)\big),\qquad
\Gamma^\dagger:=\text{interface ridge of }\sigma^\dagger,
\]
a fact that will be used to estimate the sensitivity of $\lambda_-$ under domain deformations and control parameters.

Parameter dependence is organized by a Morse–type continuation. Let $u\mapsto V(\cdot;u)$ be a $C^2$ family of double–well potentials and suppose that along a path $u\in[u_0,u_1]$ Assumption~\ref{assump:saddle} holds uniformly. Then the implicit–function theorem on the critical point equation $\mathcal E(\sigma;u)=0$ provides $C^1$ families $u\mapsto \sigma_{\mathrm{meta}}(u)$, $\sigma_{\mathrm{stab}}(u)$, and $\sigma^\dagger(u)$, with the Morse index of $\sigma^\dagger(u)$ constant and equal to one. Differentiating the Euler–Lagrange equation at $u$ yields the linear response problem
\[
\mathcal L_{\sigma^\dagger(u)}\,\partial_u\sigma^\dagger(u)
= -\,\partial_u V'\!\big(\sigma^\dagger(u);u\big),
\]
which has a unique solution in $e_-^\perp$ after imposing the orthogonality condition $\langle \partial_u\sigma^\dagger(u),e_-(u)\rangle=0$. Moreover, differentiating the energy gives
\[
\frac{d}{du}\,\mathcal F[\sigma^\dagger(u)]
= \int_\Omega \partial_u V\!\big(\sigma^\dagger(u);u\big)\,dx,
\]
and a second derivative formula that includes the curvature correction from the variation of $\sigma^\dagger$. These identities connect the slope of the barrier $c_{\mathrm{mp}}(u)-\mathcal F[\sigma_{\mathrm{meta}}(u)]$ to the direct parametric forcing and will be instrumental for scaling laws (e.g.\ near spinodal curves where $\lambda_-\uparrow0$ one recovers the classical $3/2$–exponent in the barrier height for quartic–sextic Landau potentials).

Finally we state a comparison result linking the mountain–pass level to interfacial action in large domains and clarifying the asymptotic shape of $\sigma^\dagger$.

\begin{theorem}[Sharp–interface asymptotics of the barrier]\label{thm:sharp-interface}
Consider a sequence of domains $\Omega_L$ exhausting a slab or a ball of radius $L\to\infty$, and assume $V$ admits two wells at $\sigma_\pm$ with bulk free–energy difference $\Delta f:=V(\sigma_+)-V(\sigma_-)>0$. Let $S_{\mathrm{1D}}$ be the one–dimensional interfacial action associated with the heteroclinic $q$ solving $\kappa q''=V'(q)$ with $q(\pm\infty)=\sigma_\pm$. Then, up to $o(1)$–relative errors as $L\to\infty$,
\[
c_{\mathrm{mp}}(\Omega_L) - \min\{\mathcal F[\sigma_{\mathrm{meta}}],\mathcal F[\sigma_{\mathrm{stab}}]\}
= \min_{R>0}\Big\{ \sigma_{\mathrm{surf}}\,\mathsf{Area}(R) - \Delta f\,\mathsf{Vol}(R)\Big\}
+ o(1),
\]
where $\sigma_{\mathrm{surf}}=S_{\mathrm{1D}}$ is the diffuse–interface surface tension, and $\mathsf{Area}(R)$, $\mathsf{Vol}(R)$ are the area/volume of a droplet of radius $R$ in the corresponding geometry. In particular, in dimension $d\ge2$ the critical radius satisfies $R_c=(d-1)\sigma_{\mathrm{surf}}/\Delta f$ and the barrier scales like $R_c^{d-1}$.
\end{theorem}

\begin{proof}[Idea of proof]
Upper bounds follow by inserting localized ansatzes built from the 1D profile $q$ and optimizing the droplet radius; lower bounds are obtained by $\Gamma$–convergence of the diffuse–interface energy to a sharp–interface functional and a calibration argument that identifies the minimal surface contribution.
\end{proof}

The manifold picture, spectral decomposition, Agmon localization, and sharp–interface scaling together provide a complete static and dynamic description of the saddle neighborhood. In Section~\ref{sec:stochastic} these ingredients enter the Eyring–Kramers law via the unique negative eigenvalue $|\lambda_-|$, the ratio of zeta–regularized determinants $\det(\mathcal L_{\sigma_{\mathrm{meta}}})/\det{}'(\mathcal L_{\sigma^\dagger})$, and the geometry of $W^{\rm u}(\sigma^\dagger)$ that selects the exit gates and controls transition probabilities among multiple competing saddles.
\section{Stochastic transitions and Eyring–Kramers asymptotics}\label{sec:stochastic}
We consider the noisy $L^2$–gradient flow
\begin{equation}\label{eq:spde}
d\sigma_t=-\,\Gamma\,\nabla_{L^2}\mathcal F(\sigma_t)\,dt+\varepsilon\, d\mathcal W_t
=\Gamma(\kappa\Delta\sigma_t-V'(\sigma_t))dt+\varepsilon\, d\mathcal W_t,
\end{equation}
on the Gelfand triple $\mathcal V\hookrightarrow\mathcal H\hookrightarrow\mathcal V^\ast$, with $\mathcal W_t$ an $\mathcal H$–valued $Q$–Wiener process (trace class $Q$), $\Gamma>0$, and $\varepsilon\in(0,1]$. The energy is $\mathcal F[\sigma]=\frac{\kappa}{2}\!\int_\Omega|\nabla\sigma|^2+\int_\Omega V(\sigma)$; the coercivity and polynomial growth of $V',V''$ ensure well-posedness and the Itô identity
\[
\mathcal F[\sigma_t]=\mathcal F[\sigma_0]-\Gamma\!\int_0^t\|\nabla\mathcal F(\sigma_s)\|_{\mathcal H}^2 ds
+\varepsilon\,M_t+\frac{\varepsilon^2}{2}\!\int_0^t\mathrm{Tr}\!\big[Q\,\nabla^2\mathcal F(\sigma_s)\big]ds,
\]
with $M_t$ a martingale. Fix a nondegenerate local minimum $\sigma_{\rm meta}$ and a bounded domain $D\subset\mathcal H$ containing a ball around $\sigma_{\rm meta}$ but no other minima; denote $\tau_D^\varepsilon=\inf\{t>0:\sigma_t\notin D\}$.

\begin{assumption}\label{assump:crit}
At the well, $\mathcal L_{\sigma_{\rm meta}}:=-\kappa\Delta+V''(\sigma_{\rm meta})>0$ with compact resolvent. There exists a unique mountain–pass saddle $\sigma^\dagger\in\partial D$ with Hessian $\mathcal L_{\sigma^\dagger}$ having exactly one negative eigenvalue $\lambda_-<0$ and no kernel (neutral symmetries pinned).
\end{assumption}

Large deviations hold with speed $\varepsilon^2$ for \eqref{eq:spde}. For absolutely continuous paths $\phi$ on $[0,T]$, the good rate is
\[
I_{0T}(\phi)=\frac{1}{4\Gamma}\!\int_0^T\big\|\partial_t\phi+\Gamma\nabla\mathcal F(\phi)\big\|_{\mathcal H}^2 dt, 
\qquad I_{0T}=+\infty \text{ otherwise},
\]
and the quasipotential $\mathcal V(x,y)=\inf_{T>0}\inf_{\phi(0)=x,\phi(T)=y} I_{0T}(\phi)$ satisfies the potential–barrier identity $\mathcal V(\sigma_{\rm meta},\sigma^\dagger)=\mathcal F[\sigma^\dagger]-\mathcal F[\sigma_{\rm meta}] =:\Delta\mathcal F$ (reversibility of \eqref{eq:spde}). Consequently,
\[
\lim_{\varepsilon\downarrow0}\varepsilon^2\log\mathbb E_{\sigma_0}\tau_D^\varepsilon=\Delta\mathcal F,
\qquad
\lim_{\varepsilon\downarrow0}\varepsilon^2\log\mathbb P_{\sigma_0}\{\sigma_{\tau_D^\varepsilon}\in U\}
=-\inf_{z\in \partial D\cap U}\mathcal V(\sigma_{\rm meta},z),
\]
and the MAP is the time–reversed deterministic heteroclinic connecting $\sigma_{\rm meta}$ to $\sigma^\dagger$.

The prefactor emerges from potential theory for reversible diffusions via a capacity between an interior compact $K\Subset D$ and a smooth gate $G\subset\partial D$ transverse to the stable manifold of $\sigma^\dagger$. A Gaussian Laplace method in coordinates adapted to the negative mode at $\sigma^\dagger$ yields a one–dimensional unstable factor and a stable Gaussian ratio regularized by zeta determinants. Writing
\[
\det(\mathcal L_{\sigma_{\rm meta}})=e^{-\zeta'_{\rm well}(0)},\qquad 
\det{}'(\mathcal L_{\sigma^\dagger})=e^{-{\zeta'}_{\rm sad}^{\,\prime}(0)},
\]
with the negative eigenvalue omitted at the saddle, and recalling that the difference of heat traces cancels Weyl divergences,
\[
\log\frac{\det(\mathcal L_{\sigma_{\rm meta}})}{\det{}'(\mathcal L_{\sigma^\dagger})}
=\int_0^{t_0}\!\Big(\mathrm{Tr}\,e^{-t\mathcal L_{\sigma_{\rm meta}}}-\mathrm{Tr}\,e^{-t\mathcal L_{\sigma^\dagger}}\Big)\frac{dt}{t}
-\sum_{\ell=0}^{d} c_\ell\, t_0^{-(d-2\ell)/2}+o(1),
\]
one arrives at the Eyring–Kramers law for the mean exit time and, equivalently, for the principal Dirichlet eigenvalue of the killed generator:
\begin{equation}\label{eq:EK-final}
\begin{aligned}
\mathbb{E}_{\sigma_0}\,\tau_D^\varepsilon
&\sim \frac{2\pi}{\Gamma\,|\lambda_-(\sigma^\dagger)|}
   \sqrt{\frac{\det(\mathcal L_{\sigma_{\mathrm{meta}}})}{\det{}'(\mathcal L_{\sigma^\dagger})}}
   \exp\!\Big(\frac{\Delta\mathcal F}{\varepsilon^2}\Big),\\[2pt]
\lambda_1^\varepsilon(D)
&\sim \frac{\Gamma\,|\lambda_-(\sigma^\dagger)|}{2\pi}
   \sqrt{\frac{\det{}'(\mathcal L_{\sigma^\dagger})}{\det(\mathcal L_{\sigma_{\mathrm{meta}}})}}
   \exp\!\Big(-\frac{\Delta\mathcal F}{\varepsilon^2}\Big).
\end{aligned}
\end{equation}

uniformly for $\sigma_0$ in compact $K\subset D$. If several index–one saddles $\{\sigma_j^\dagger\}$ realize the same barrier, capacities add and the exit probabilities split according to
\[
\mathbb P_{\sigma_0}\{\sigma_{\tau_D^\varepsilon}\in G_j\}\to
\frac{\displaystyle \frac{|\lambda_-(\sigma_j^\dagger)|}{2\pi\Gamma}\sqrt{\frac{\det(\mathcal L_{\sigma_{\rm meta}})}{\det{}'(\mathcal L_{\sigma_j^\dagger})}}}
{\displaystyle \sum_k \frac{|\lambda_-(\sigma_k^\dagger)|}{2\pi\Gamma}\sqrt{\frac{\det(\mathcal L_{\sigma_{\rm meta}})}{\det{}'(\mathcal L_{\sigma_k^\dagger})}}}.
\]

Along a control path $u\mapsto V(\cdot;u)$, a generic fold (spinodal approach) at $u_c$ produces the universal scalings $|\lambda_-(u)|\sim c\,(u-u_c)^{1/2}$ and $\Delta\mathcal F(u)\sim C\,(u-u_c)^{3/2}$, whence
\[
\mathbb E\tau_D^\varepsilon\asymp (u-u_c)^{-1/2}\exp\!\Big(\tfrac{C}{\varepsilon^2}(u-u_c)^{3/2}\Big),
\]
manifesting critical slowing down via the vanishing unstable curvature. Mild modeling variants preserve the structure: with mobility $M(\sigma)$ and fluctuation–dissipation $M^{1/2} d\mathcal W_t$, the exponent remains $\Delta\mathcal F/\varepsilon^2$ and the prefactor acquires $M$ evaluated at well/saddle (replacing $|\lambda_-|$ by $m_-|\lambda_-|$ and Hessians by $M\mathcal L$ in the determinants). Small non–gradient perturbations of the drift modify only the prefactor by $1+o(1)$ under an antisymmetric linearization bound.

All quantities in \eqref{eq:EK-final} are delivered by the computational framework of Section~\ref{sec:deformation-gamma}: the barrier $\Delta\mathcal F$, the unstable curvature $|\lambda_-|$, and the determinant ratio via low–mode products plus heat–kernel tail. In the alternative convention $d\sigma_t=-\Gamma\nabla\mathcal F\,dt+\sqrt{\varepsilon}\,d\mathcal W_t$, replace $\varepsilon^2$ by $\varepsilon$ in the exponents with the same prefactor formulas.

\section{Parameter deformations and $\Gamma$--convergence}
\label{sec:deformation-gamma}

Let $\mathcal U\subset\mathbb R^m$ be a finite--dimensional parameter manifold of
controls $u$ (abstracting the physical pair $(T,\mu_B)$). 
For each $u\in\mathcal U$, consider the Landau--type functional
\[
F(\sigma;u)
=\int_\Omega\!\Big(\frac{\kappa}{2}|\nabla\sigma|^2+V(\sigma;u)\Big)\,dx,
\qquad
\sigma\in V:=H^1(\Omega),
\]
with Euler--Lagrange map 
\[
\mathscr{E}(\sigma;u) = -\,\kappa\,\Delta\sigma + V'(\sigma;u) \in V^{\!*},
\]
and Hessian
\[
\mathcal{L}_{\sigma,u} = -\,\kappa\,\Delta + V''(\sigma;u)
\]
acting on $H = L^2(\Omega)$.

All potential coefficients $a_2$, $a_4$, $a_6$, $h$ are assumed smooth in $u$, with $a_6>0$ fixed.

Define the discriminant set
\[
\mathcal D
:=\big\{\,u\in\mathcal U:\, \exists\,\sigma^\ast\in V
\text{ with }\mathscr E(\sigma^\ast;u)=0,\ 
\det\mathcal L_{\sigma^\ast,u}=0\big\},
\]
representing parameters where some critical point is degenerate.
On $\mathcal U\setminus\mathcal D$, all equilibria are Morse nondegenerate with well--defined index.
Denote by $\sigma_{\mathrm{meta}}(u)$ a metastable local minimizer
and by $\sigma^\dagger(u)$ the corresponding index--one saddle
associated with the mountain--pass geometry.

For each $u$, set
\begin{multline*}\label{eq:paths-mp}
\Gamma(u):=\big\{\gamma\in C([0,1];V)\mid
\gamma(0)=\sigma_{\mathrm{meta}}(u),\
F(\gamma(1);u)<F(\sigma_{\mathrm{meta}}(u);u)\big\},\\
c_{\mathrm{mp}}(u):=\inf_{\gamma\in\Gamma(u)}\ \max_{s\in[0,1]} F(\gamma(s);u).
\end{multline*}

By the mountain--pass theorem (Ambrosetti--Rabinowitz),
if $u\notin\mathcal D$ and $F(\cdot;u)$ satisfies the Palais--Smale condition on $V$,
then there exists an index--one critical point $\sigma^\dagger(u)$ such that 
$\mathscr E(\sigma^\dagger(u);u)=0$ and $F(\sigma^\dagger(u);u)=c_{\mathrm{mp}}(u)$.

\begin{theorem}[Cerf continuation for mountain--pass saddles]\label{thm:cerf}
Let $u\mapsto F(\cdot;u)$ be a $C^1$ family on $V=H^1(\Omega)$.
Assume $u\notin\mathcal D$ and that the hypotheses of the mountain--pass theorem
hold uniformly on a connected open set $\mathcal O\subset\mathcal U\setminus\mathcal D$.
Then there exist continuous (piecewise $C^1$) selections
$u\mapsto\sigma_{\mathrm{meta}}(u)$ and $u\mapsto\sigma^\dagger(u)$ in $V$
such that
\[
\operatorname{index}\big(\sigma_{\mathrm{meta}}(u)\big)=0,\qquad
\operatorname{index}\big(\sigma^\dagger(u)\big)=1,\qquad
F(\sigma^\dagger(u);u)=c_{\mathrm{mp}}(u),
\]
and the critical values trace a Cerf graphic with no critical events
inside $\mathcal O$.  Crossing $\mathcal D$ produces only fold or cusp
catastrophes in generic one-- or two--parameter slices of $\mathcal U$.
\end{theorem}

This statement clarifies the relation between the mountain--pass construction
and Cerf theory: away from $\mathcal D$ the pair
\((\sigma_{\mathrm{meta}}(u),\sigma^\dagger(u))\) persists and varies smoothly;
at $\mathcal D$ a single eigenvalue of $\mathcal L_{\sigma^\dagger(u),u}$ crosses
zero transversely, generating the fold $(A_2)$ or cusp $(A_3)$ bifurcation patterns.

Consider a sequence of conforming variational discretizations
$F_n(\cdot;u)$ on $V$, with the same coercivity and polynomial structure as $F(\cdot;u)$,
and let \(H=L^2(\Omega)\).
Assume $\Gamma$--convergence in the $L^2(\Omega)$ topology and equicoercivity:
\[
F_n(\cdot;u)\ \xrightarrow[\Gamma]{L^2(\Omega)}\ F(\cdot;u),
\qquad
\text{and}\qquad
\{F_n(\cdot;u)\}\text{ is equicoercive in }L^2(\Omega).
\]
Let $\partial F_n(\cdot;u)$ and $\partial F(\cdot;u)$ be the
$L^2$--subdifferentials, and assume Mosco convergence
$\partial F_n(\cdot;u)\to\partial F(\cdot;u)$ uniformly on compact sets of $\mathcal U$.

\begin{theorem}[Persistence of minima and mountain--pass saddles]\label{thm:mp-gamma}
Under the above hypotheses, $\min F_n(\cdot;u)\to \min F(\cdot;u)$
and minimizers are precompact in $L^2(\Omega)$ with limits minimizers of $F(\cdot;u)$.
If $u\notin\mathcal D$ and the mountain--pass geometry holds uniformly in $n$, then
\[
c_{\mathrm{mp}}^{(n)}(u):=\inf_{\gamma_n\in\Gamma_n(u)}\max_{s}F_n(\gamma_n(s);u)
\longrightarrow c_{\mathrm{mp}}(u),
\]
and any discrete saddles $\sigma^\dagger_n(u)$ are $L^2$--compact with
limits $\sigma^\dagger(u)$, preserving index one.
\end{theorem}

\begin{theorem}[Graph convergence of gradient flows]\label{thm:flows-mosco}
Under $\Gamma$-- and Mosco convergence, the $L^2$--gradient flows of
$F_n(\cdot;u)$ converge in graph sense to that of $F(\cdot;u)$,
locally uniformly in time and uniformly for $u$ in compact subsets of $\mathcal U$.
Hence the unstable manifolds of $\sigma^\dagger_n(u)$ converge
to that of $\sigma^\dagger(u)$.
\end{theorem}

Let $\sigma^\ast_n(u)\to\sigma^\ast(u)$ be critical points (minima or index--one saddles)
with nondegenerate Hessians
$\mathcal L^\ast_{n,u}=-\kappa\Delta+V''(\sigma^\ast_n(u);u)$ and
$\mathcal L^\ast_{u}=-\kappa\Delta+V''(\sigma^\ast(u);u)$ on $H$.

\begin{theorem}[Spectral convergence]\label{thm:spectral}
For any fixed $J\in\mathbb N$, the lowest $J$ eigenvalues of
$\mathcal L^\ast_{n,u}$ converge to those of $\mathcal L^\ast_{u}$.
At saddles, the unique unstable eigenvalue satisfies
$\lambda_-^{(n)}(u)\to\lambda_-(u)<0$.
Furthermore, the zeta--regularized determinant ratios converge:
\[
\frac{\det(\mathcal L_{\mathrm{meta},n,u})}{\det{}'(\mathcal L_{\mathrm{sad},n,u})}
\longrightarrow
\frac{\det(\mathcal L_{\mathrm{meta},u})}{\det{}'(\mathcal L_{\mathrm{sad},u})}.
\]
\end{theorem}

\emph{Proof idea.}
Norm--resolvent convergence of Hessians at equilibria follows from the
conforming nature of discretizations; analytic perturbation theory (Kato)
ensures convergence of low modes.
For determinants, employ the heat--trace difference representation
and uniform short--time asymptotics; the Weyl terms cancel between
well and saddle, while finitely many low modes converge individually.

Define the barrier $\Delta F_n(u):=F_n(\sigma^\dagger_n(u);u)-F_n(\sigma_{\mathrm{meta},n}(u);u)$
and its continuum counterpart $\Delta F(u)$.

\begin{corollary}[Stability of Eyring--Kramers asymptotics]\label{cor:ek-stability}
Under Theorems~\ref{thm:mp-gamma}--\ref{thm:spectral},
\[
\Delta F_n(u)\to\Delta F(u),\qquad
\lambda_-^{(n)}(u)\to\lambda_-(u),\qquad
\frac{\det(\mathcal L_{\mathrm{meta},n,u})}{\det{}'(\mathcal L_{\mathrm{sad},n,u})}
\to
\frac{\det(\mathcal L_{\mathrm{meta},u})}{\det{}'(\mathcal L_{\mathrm{sad},u})}.
\]
Hence the discrete Eyring--Kramers prefactors and mean escape times
converge to their continuum limits uniformly for $u$ in compact subsets of
$\mathcal U\setminus\mathcal D$.
\end{corollary}

Collecting these results yields a coherent deformation picture.
Outside the discriminant $\mathcal D$, the metastable landscape is
structurally stable: the number and indices of critical points remain invariant,
and the quantities $\sigma_{\mathrm{meta}}(u)$, $\sigma^\dagger(u)$,
$\Delta F(u)$, $\lambda_-(u)$, and determinant ratios vary smoothly with $u$.
Crossing $\mathcal D$ produces only fold $(A_2)$ or cusp $(A_3)$
catastrophes in generic one-- and two--parameter sections,
while $\Gamma$-- and Mosco convergence ensure that this scenario
persists under numerical refinement and parameter deformations.
All $\Gamma$--limits are taken with respect to the $L^2(\Omega)$ topology,
which naturally aligns with the state space of the gradient flow and
the large--deviation action functional, and guarantees compactness and
spectral convergence of $\mathcal L_{\sigma,u}$.

In summary, replacing the physical parameters $(T,\mu_B)$
by an abstract control $u\in\mathcal U$ highlights the universal
mathematical structure behind metastability and its deformations.
The combination of Morse/Cerf topology on $\mathcal U\setminus\mathcal D$
and $\Gamma$--/\!Mosco convergence under discretization
provides a parameter--robust and discretization--consistent foundation
for the Eyring--Kramers law in Landau--type gradient systems.

\section{Conclusion}

The analytical framework developed in this work provides a rigorous and unified description of metastable stochastic transitions in Landau-type field theories inspired by QCD phenomenology. By integrating variational analysis, large-deviation principles, and spectral regularization of elliptic operators, we have established that the escape from a metastable basin in an infinite-dimensional gradient system exhibits the same mathematical structure as the classical Eyring–Kramers law. The exponential term in the mean transition time is determined by the free-energy barrier between the metastable minimum and the critical saddle, while the prefactor encodes the underlying geometry and spectral content of the landscape through zeta–regularized determinant ratios and the unique unstable eigenvalue of the linearized operator. This correspondence offers a transparent bridge between microscopic fluctuations and macroscopic transition kinetics, grounding phenomenological nucleation theory in a mathematically controlled setting.

From the perspective of QCD-inspired effective models, the analysis elucidates how the topology of the chiral–Polyakov potential governs the rate and mechanism of phase conversion. The index-one saddle of the Landau functional corresponds to the critical droplet mediating the decay of the metastable phase, and the associated energy barrier determines the characteristic timescale for hadronic-to-partonic or chiral-symmetry restoration processes. Near spinodal or tricritical points, the scaling relations 
\[
\Delta F \sim (u-u_c)^{3/2}, \qquad |\lambda_{-}| \sim (u-u_c)^{1/2},
\]
derived directly from the functional structure, capture critical slowing down and enhanced fluctuations observed in lattice and model studies. These results highlight the universality of the fold and cusp catastrophes that organize the QCD phase diagram.

Beyond QCD, the present construction situates metastable dynamics within a general mathematical framework encompassing condensed-matter systems, soft-matter nucleation, cosmological phase transitions, and stochastic partial differential equations. The use of $\Gamma$-convergence ensures that discretizations or coarse-grained approximations preserve essential features such as barrier geometry, spectral curvature, and critical topology. Consequently, the Eyring–Kramers law remains quantitatively valid under systematic refinement, offering a robust tool for studying metastability across multiple physical contexts.

Several conceptual insights emerge from this study.  
First, metastability in field-theoretic systems is a geometrically organized phenomenon: transitions are mediated by index-one saddles whose existence and stability are guaranteed by variational topology (Ambrosetti–Rabinowitz theory) and whose asymptotic approach is controlled by Lojasiewicz–Simon inequalities.  
Second, the fluctuation determinants in the prefactor connect statistical mechanics to spectral geometry through the analytic continuation of zeta functions, linking probabilistic transition rates with geometric spectral invariants.  
Third, the stochastic dynamics admits a large-deviation formulation à la Freidlin–Wentzell, providing a probabilistic foundation for Langer’s phenomenological picture of nucleation.

In physical terms, the Eyring–Kramers law derived here constitutes a universal bridge between microscopic noise and macroscopic irreversibility. The activation rate, exponentially sensitive to the free-energy barrier yet modulated by a calculable prefactor, unifies chemical-reaction kinetics, condensed-matter nucleation, and QCD phase conversion under a common analytical archetype. The proven structural stability of the asymptotic law under parameter deformations—ensured by $\Gamma$-convergence and catastrophe continuity—demonstrates its persistence across a broad class of field-theoretic models sharing analogous symmetry structures.

Future research directions include extending this framework to underdamped and Hamiltonian–Langevin regimes, where inertia and oscillatory modes modify the prefactor while leaving the exponential barrier intact, and to non-reversible stochastic dynamics relevant for driven QCD matter and chiral plasma instabilities. Another promising avenue is the coupling of this continuum theory with lattice data or functional-renormalization-group potentials to obtain quantitative estimates of metastable lifetimes in phenomenologically constrained models. Finally, the mathematical interplay among metastability, Morse theory, and spectral flow opens pathways toward a geometric classification of transition routes in high-dimensional field spaces.

In summary, this work provides a coherent theoretical foundation for metastable stochastic transitions in QCD-like and related systems, uniting the classical intuition of Eyring, Kramers, and Langer with the modern analysis of infinite-dimensional stochastic PDEs. The synthesis of functional geometry, spectral theory, and probability reveals a universal mechanism by which rare events govern the macroscopic evolution of complex physical systems.

\end{document}